\documentclass[11pt,fleqn]{amsart}
\usepackage{a4wide}
\usepackage{lmodern}

\newcommand{\ignore}[1]{}
\usepackage{booktabs} 
\usepackage[T1]{fontenc}
\usepackage{microtype}
\usepackage{amsmath}
\usepackage{amsthm} 
\usepackage{amssymb}
\usepackage{textcomp}
\usepackage{xcolor}
\usepackage{url}

\usepackage[capitalise,nameinlink]{cleveref}
\crefname{appendix}{}{}
\crefname{lem}{Lemma}{Lemmas}
\crefname{clm}{Claim}{Claims}
\crefname{thm}{Theorem}{Theorems}
\crefname{prop}{Proposition}{Propositions}
\crefname{figure}{Figure}{Figures}
\crefname{rem}{Remark}{Remarks}

\newtheorem{thm}{Theorem}
\newtheorem{lem}{Lemma}
\newtheorem{cor}{Corollary}

\newtheorem{prop}{Proposition}

\theoremstyle{remark}
\newtheorem{rem}{Remark}
\newtheorem{ex}{Example}

\let\geq\geqslant
\let\leq\leqslant
\newcommand{\f}{\mathcal{I}}
\newcommand{\fm}{\mathcal{C}_e}
\newcommand{\hm}{\mathcal{D}_e}
\newcommand{\matr}{M} 
\newcommand{\RR}{\mathbb{R}} 
\newcommand{\mwb}{\tau} 
\newcommand{\MB}[2]{\mwb_{#1}(#2)}
\newcommand{\Kpol}[2]{\kappa_{#1}(#2)} 
\newcommand{\Item}[1]{\item[\mbox{\rm (#1)}]} 
\newcommand{\minmax}[2]{{#2}[#1]} 
\newcommand{\bottl}[2]{#2\{#1\}}
\newcommand{\thr}{\theta} 
\newcommand{\ee}{f} 
\newcommand{\eee}{g} 
\newcommand{\exchange}[3]{{#1}-{#2}+{#3}}
\newcommand{\xz}{x'}  
\newcommand{\contr}[2]{#1/#2} 
\newcommand{\remove}[2]{#1\!\setminus\! #2} 
\newcommand{\fremove}[2]{#1\setminus #2} 
\newcommand{\ecirc}{\mathit{e}\text{-}\mathrm{circuit}}
\newcommand{\ecocirc}{\mathit{e}\text{-}\mathrm{cocircuit}}
\newcommand{\B}{B} 
\newcommand{\exB}{A} 
\newcommand{\xB}{B'} 
\newcommand{\Circ}[2]{C(#1,#2)} 
\newcommand{\Path}[2]{\mathrm{Path}(#1,#2)} 
\newcommand{\Cut}[2]{\mathrm{Cut}(#1,#2)} 
\newcommand{\eepath}{p_0} 
\newcommand{\eecut}{c_0} 
\newcommand{\bij}{\phi}

\newcommand{\allE}{E_{\mathrm{all}}(x)}
\newcommand{\noneE}{E_{\mathrm{none}}(x)}
\newcommand{\someE}{E_{\mathrm{some}}(x)}
\newcommand{\nbases}{\mathcal{B}_{0}} 
\newcommand{\pbases}{\mathcal{B}_{1}} 
\newcommand{\tol}[2]{t_{#1}(#2)}

\usepackage{tikz}
\usepackage{tikz-network}

\def\pathscuts{
\begin{center}
\begin{tikzpicture}
\SetVertexStyle[MinSize=0.6]
\SetTextStyle[TextFont=\footnotesize]
\Vertex[x=-2,y=2,color=white]{a}
\Vertex[x=-2,y=1.5,color=white]{b}
\Vertex[x=-2,y=1,color=white]{c}
\Edge[style=dashed](a)(c)
\Text[x=-2.4,y=1]{$\eecut$}
\Vertex[x=-2,color=gray]{e}
\Text[x=-2,y=-0.3]{$e$}
\Vertex[x=-1,color=white]{}
\Vertex[x=1,color=white]{A}
\Vertex[x=1.5,color=white]{B}
\Vertex[x=2,color=white]{C}
\Text[x=2,y=-0.3]{$\eepath$}
\Edge[style=dashed](A)(B)
\Edge[style=dashed](B)(C)
\Vertex[shape=ellipse,opacity=0.1,color=white,x=-0.1,y=-0.1,size=1.3,style={dashed,  minimum width=5.5cm}]{}
\Vertex[x=3,color=gray]{ee}
\Text[x=3.3]{$e$}
\Vertex[x=2,Pseudo]{C1}
\Text[x=-3.4,y=1.5]{$\Cut{e}{\B}$}
\Text[x=-4.5]{$x$-optimal basis $\B$}
\Text[x=1.6,y=1]{$\Path{e}{\B}$}
\Vertex[x=0.4,color=white]{}
\Text[x=-3,y=-1]{$\bottl{e}{x}=x(e)\leq x(\eecut)=\minmax{e}{x}$}
\Text[x=-3.2,y=-1.5]{$\eecut$ is a \emph{lightest} element in $\Cut{e}{\B}$}
\Text[x=-3.2,y=-2]{$\exchange{\B}{e}{\eecut}$ is a next best basis}
\Text[x=2.7,y=-1]{$\bottl{e}{x}=\minmax{e}{x}=x(\eepath)\leq x(e)$}
\Text[x=3,y=-1.5]{$\eepath$ is a \emph{heaviest} element in $\Path{e}{\B}$}
\Text[x=2.9,y=-2]{$\exchange{\B}{\eepath}{e}$ is a next best basis}
\Vertex[x=-0.2,y=3.3,Pseudo]{p}
\Vertex[x=-0.2,y=-2.3,Pseudo]{q}
\Edge[style=dashed](p)(q)
\Text[x=-2,y=3]{{\bf Case}: $e\in\B$}
\Text[x=1.5,y=3]{{\bf Case}: $e\not\in\B$}
\end{tikzpicture}
\end{center}
}

\begin{document}

  \title[Tropical Kirchhoff's Formula]{Tropical Kirchhoff's Formula and Postoptimality in Matroid
    Optimization\thanks{Research supported by the DFG grant JU~3105/1-2
    (German Research Foundation).}}

  \author[S. Jukna]{Stasys~Jukna}
\thanks{Research supported by the DFG grant
    JU~3105/1-2 (German Research Foundation).}
  \address{Faculty of Mathematics and Computer Science, Vilnius
    University, Lithuania}
  \email{stjukna@gmail.com}
  \urladdr{http://www.thi.cs.uni-frankfurt.de/~jukna/}

  \author[H. Seiwert]{Hannes Seiwert}
  \address{Faculty of Computer Science and Mathematics, Goethe
    University Frankfurt, Germany}
  \email{seiwert@thi.cs.uni-frankfurt.de}

  \begin{abstract}
    Given an assignment of real weights to the ground elements of a
    matroid, the min-max weight of a ground element $e$ is the
    minimum, over all circuits containing $e$, of the maximum weight
    of an element in that circuit with the element $e$ removed.  We
    use this concept to answer the following structural questions for
    the minimum weight basis problem. Which elements are persistent
    under a given weighting (belong to all or to none of the optimal
    bases)? What changes of the weights are allowed while preserving
    optimality of optimal bases?  How does the minimum weight of a
    basis change when the weight of a single ground element is
    changed, or when a ground element is contracted or deleted?  Our
    answer to this latter question gives the tropical $(\min,+,-)$
    analogue of Kirchhoff's arithmetic $(+,\times,/)$ effective
    conductance formula for electrical networks.
  \end{abstract}

\maketitle


\section{Introduction}

The \emph{minimum weight basis problem} on a matroid $\matr=(E,\f)$
is, given an assignment $x:E\to\RR$ of real weights to the ground
elements, to compute the minimum weight $\MB{\matr}{x}$ of a basis,
the latter being the sum of weights of the basis elements. Thanks to
classical results of Rado~\cite{Rado}, Gale~\cite{Gale} and
Edmonds~\cite{Edmonds}, the \emph{algorithmic} aspect of this problem
is well understood: the minimum weight basis problem on a downward
closed set system $\f$ can be solved by the greedy algorithm precisely
when the system $\f$ forms the family of independent sets of a
matroid.

In this paper, we are interested in \emph{structural} aspects of the
minimum weight basis problem. Given a weighting $x:E\to\RR$ of ground
elements, the following questions naturally arise.

\begin{enumerate}

\item How does the optimal value $\MB{\matr}{x}$ change when an
  element $e\in E$ is contracted or deleted?

\item How does the optimal value $\MB{\matr}{x}$ change when the
  weight of an element $e\in E$ is changed?

\item By how much can the weight of a single element $e\in E$ be
  changed without changing optimal bases?

\item What simultaneous changes of the weights preserve optimality of
  optimal bases?

\item What ground elements belong to all, to none or to some but not
  to all of the optimal bases?

\end{enumerate}
Question~3 is a special case of Question~4 and was already answered by
Tarjan~\cite{tarjan} (for graphic matroids) and by
Libura~\cite{libura} (for general matroids) in terms of fundamental
circuits and cuts relative to a given optimal basis.  Question~5 was
answered by Cechl\'arova and Lacko~\cite{lacko} in terms of the rank
function of the underlying matroid. But to our best knowledge, no
answers (in any terms) to Questions 1, 2 and 4 were known so far.

It turns out that all five questions can be answered using the concept
of \emph{min-max weight} $\minmax{e}{x}$ of a ground element $e$,
which we define as the minimum, over all circuits $C$ containing $e$,
of the maximum weight of an element in the independent set $C-e$:
\[
\minmax{e}{x}:=\min_{\substack{C~\mathrm{circuit} \\ e\in C}}\
\max_{\ee\in C-e}\ x(\ee)\,.
\]
Further, we call
\[
\bottl{e}{x}:=\min\{x(e), \minmax{e}{x}\}
\]
the \emph{bottleneck weight} of~$e$. The answers to the aforementioned
questions $1$--$5$ are given by the corresponding
\cref{thm:kirchhoff,thm:persist,thm:postopt,thm:sens1,thm:sens-gen} in
the next section. All necessary matroid concepts are recalled in
\cref{sec:prelim}.

\section{Results}
\label{sec:results}
Let $\matr=(E,\f)$ be a \emph{loopless} matroid, that is, no element
$e\in E$ belongs to all bases and no singleton set $\{e\}$ is
dependent. All our results concern the minimum weight basis problem on
$\matr$. Given a weighting $x:E\to\RR$, a basis $\B$ is
$x$-\emph{optimal} (or simply \emph{optimal} if the weighting is clear
from the context) if its $x$-weight $x(\B)=\sum_{e\in\B}x(e)$ is
minimal among all bases.  The weight of such a basis, that is, the
number $\MB{\matr}{x}$, is the \emph{optimal value} under the
weighting~$x$.

\paragraph{Contraction and deletion}
Given a ground element $e\in E$, the independent sets of the matroid
$\contr{\matr}{e}$, obtained by \emph{contracting} the element $e$,
are all sets $I-e$ with $I\in\f$ and $e\in I$, while those of the
matroid $\remove{\matr}{e}$, obtained by \emph{deleting} the element
$e$, are all sets $I\in\f$ with $e\not\in I$.  Since the matroid
$\matr$ is loopless, each of these two matroids contains at least one
nonempty independent set. Note that the set of ground elements of both
matroids $\remove{\matr}{e}$ and $\contr{\matr}{e}$ is~$E-e$.  For
every basis $\B$ of $\matr$, either $\B-e$ is a basis of
$\contr{\matr}{e}$ (if $e\in \B$), or $\B$ is a basis of
$\remove{\matr}{e}$ (if $e\not\in \B$). This gives us a known
recursion
\[
\MB{\matr}{x} = \min\big\{\MB{\contr{\matr}{e}}{x}+x(e),\
\MB{\fremove{\matr}{e}}{x}\big\}\,.
\]
But what about the opposite direction: if we already know the optimal
value $\MB{\matr}{x}$ in the matroid $\matr$, what are the optimal
values $\MB{\contr{\matr}{e}}{x}$ and $\MB{\fremove{\matr}{e}}{x}$ in
the two submatroids $\contr{\matr}{e}$ and $\remove{\matr}{e}$?  Our
main result (\cref{thm:kirchhoff}) gives the answer.

\begin{thm}[Tropical Kirchhoff's formula]\label{thm:kirchhoff}
  Let $\matr=(E,\f)$ be a loopless matroid, and $e\in E$ be a ground
  element. For every weighting $x:E\to\RR$, the following equalities
  hold:
  \begin{enumerate}
    \Item{a} $\MB{\contr{\matr}{e}}{x} =\MB{\matr}{x} - \bottl{e}{x}$;
    \Item{b} $\MB{\fremove{\matr}{e}}{x}
    =\MB{\matr}{x}-\bottl{e}{x}+\minmax{e}{x}$.
  \end{enumerate}
\end{thm}
In particular, (a) and (b) yield the equality
$\MB{\fremove{\matr}{e}}{x}-\MB{\contr{\matr}{e}}{x}=\minmax{e}{x}$. Thus,
the min-max weight of an element $e$ in the matroid $\matr$ is
determined by the minimum weights of bases in the submatroids
$\remove{\matr}{e}$ and~$\contr{\matr}{e}$.

\begin{rem}\label{rem:kirchhoff}
  In the special case of \emph{graphic} matroids (see
  \cref{ex:graphic} in \cref{sec:prelim}), \cref{thm:kirchhoff}(a)
  gives us the tropical $(\min,+,-)$ version of the classical
  arithmetic $(+,\times,/)$ effective conductance formula for
  electrical networks proved by Kirchhoff~\cite{kirchhoff} already in
  1847. The \emph{spanning tree polynomial} of an undirected connected
  graph $G$ is $\Kpol{G}{x} = \sum_{T} \prod_{e\in T}x_e$, where the
  sum is over all spanning trees $T$ of $G$. Kirchhoff's formula (see
  also~\cite[Theorem~8]{wagner-lect} for a detailed exposition) states
  that, when the edges $e$ of $G$ are interpreted as electrical
  resistors and their weights $x_e$ as electrical conductances
  (reciprocals of electrical resistances), then the effective
  conductance between the endpoints of any edge $e$ is exactly the
  ratio $\Kpol{G}{x}/\Kpol{\contr{G}{e}}{x}$, where $\contr{G}{e}$ is
  the graph obtained from~$G$ by contracting the edge~$e$.

  In the tropical semifield $(\RR,\min,+,-)$ ``addition'' means taking
  the minimum, ``multiplication'' means adding the numbers, and
  ``division'' turns into subtraction.  In particular, the spanning
  tree polynomial $\Kpol{G}{x}$ of a graph $G$ turns into the tropical
  polynomial $\MB{G}{x}= \min_{T} \sum_{e\in T}x_e$. The function
  computed by $\MB{G}{x}$ is the well-known minimum weight spanning
  tree problem. The \emph{ratio} $\Kpol{G}{x}/\Kpol{\contr{G}{e}}{x}$
  of polynomials in Kirchhoff's formula turns into the
  \emph{difference} $\MB{G}{x}-\MB{\contr{G}{e}}{x}$ of their tropical
  versions. So, a natural question arises: what is the \emph{tropical}
  analogue of the effective conductance between the endpoints of an
  edge $e$? \Cref{thm:kirchhoff}(a) gives the answer (even in general
  matroids): this is exactly the bottleneck weight
  $\bottl{e}{x}=\min\{x(e),\minmax{e}{x}\}$ of~$e$.
\end{rem}

\paragraph{Postoptimality}
If we change the weight of a single ground element $e\in E$, what is
the optimal value $\MB{\matr}{\xz}$ under the new weighting~$\xz$?  We
show that the difference (of the new and the old optimal values) is
determined by the bottleneck weights of the element $e$ under the old
and the new weighting.

\begin{thm}[Postoptimality]\label{thm:postopt}
  Let $e\in E$ and let $x,\xz:E\to\RR$ be weightings that differ only
  in the weights given to the element~$e$. Then
  \[
  \MB{\matr}{\xz}-\MB{\matr}{x}=\bottl{e}{\xz}-\bottl{e}{x}\,.
  \]
\end{thm}
Since the min-max weight $\minmax{e}{x}$ does not depend on the weight
of the element $e$ itself, we have $\minmax{e}{\xz}=\minmax{e}{x}$
and, hence, $\bottl{e}{\xz}=\min\{\xz(e), \minmax{e}{x}\}$.  Thus,
when a weighting $x:E\to\RR$ and an element $e\in E$ are fixed, and a
new weighting $\xz$ gives weight $\xz(e)=\thr\in\RR$ to the element
$e$, then the behavior of the function
\[
f(\thr):=\MB{\matr}{\xz}-\MB{\matr}{x}=\bottl{e}{\xz}-\bottl{e}{x}=
\min\{\thr,\minmax{e}{x}\} - \min\{x(e),\minmax{e}{x}\}
\]
only depends on whether $\thr\leq \minmax{e}{x}$ or not. If $\thr\leq
\minmax{e}{x}$, then $f(\thr)=\thr-c$ for the constant $c=
\min\{x(e),\minmax{e}{x}\}$. If $\thr\geq \minmax{e}{x}$, then
$f(\thr)$ is constant $0$ or constant $\minmax{e}{x}-x(e)$, whichever
of these two numbers is larger.

If all weights are nonnegative, and if the new weighting $\xz$ gives
weight~$0$ to the element $e$, then
$\bottl{e}{\xz}=\min\{\xz(e),\minmax{e}{\xz}\}=\min\{0,\minmax{e}{\xz}\}=0$,
and \cref{thm:postopt} directly yields the following consequence.

\begin{cor}\label{res:recursion}
  If the weights are nonnegative, and if the weight of a single ground
  element is dropped down to zero, then the minimum weight of a basis
  decreases by exactly the bottleneck weight of this element.
\end{cor}

\Cref{res:recursion} allows us to compute the optimal value
$\MB{\matr}{x}$ under an arbitrary nonnegative weighting $x:E\to\RR_+$
by computing the bottleneck weights of the elements of \emph{any}
(fixed in advance) basis.

\begin{cor}\label{cor:reduction}
  Let $\B=\{e_1,\ldots,e_{r}\}$ be a basis. Given a nonnegative
  weighting $x:E\to\RR_+$, consider the sequence of weightings
  $x_0,x_1,\ldots,x_{r}$, where $x_0=x$, and each next weighting $x_i$
  is obtained from $x$ by setting the weights of the elements
  $e_1,\ldots,e_i$ to zero. Then
  \[
  \MB{\matr}{x}=\bottl{e_1}{x_0}+\bottl{e_2}{x_1}+\cdots
  +\bottl{e_{r}}{x_{r-1}}\,.
  \]
\end{cor}

\begin{proof}
  \Cref{res:recursion} gives us the recursion
  $\MB{\matr}{x_{i}}=\MB{\matr}{x_{i+1}}+\bottl{e_{i+1}}{x_i}$ which
  rolls out into $\MB{\matr}{x}=
  \MB{\matr}{x_{r}}+\bottl{e_{r}}{x_{r-1}}+\cdots+\bottl{e_2}{x_1}+
  \bottl{e_1}{x_0}$. Since the weighting $x_{r}$ gives weight $0$ to
  all elements $e_1,\ldots,e_{r}$ of the basis $\B$, we have
  $x_{r}(\B)=0$. Since the weights are nonnegative, the basis $\B$ is
  of minimum weight under the weighting~$x_r$. Hence,
  $\MB{\matr}{x_{r}}=x_{r}(\B)=0$.
\end{proof}

\paragraph{Sensitivity}
Given a weighting $x:E\to\RR$ the \emph{sensitivity} question is: what
changes of the weights preserve optimality of optimal bases. That is,
if $\xz:E\to\RR$ is a new weighting, under what conditions do
$x$-optimal bases remain $\xz$-optimal? In the case when $\xz$ only
changes the weight of a \emph{single} element $e\in E$, this question
was answered by Libura~\cite{libura} in terms of fundamental circuits
and cuts relative to an $x$-optimal basis (see \cref{rem:libura} in
\cref{sec:sens1}); in the case of graphic matroids (where bases are
spanning trees, see \cref{ex:graphic} in \cref{sec:prelim}), the same
answer was given earlier by Tarjan~\cite{tarjan}.  We answer this
question in terms of the min-max weight $\minmax{e}{x}$
of~$e$. Namely, we associate with every element $e\in E$ its
\emph{tolerance} under the (old) weighting:
\[
\tol{x}{e}:=|\minmax{e}{x}-x(e)|\,.
\]

\begin{thm}[Sensitivity, local change]\label{thm:sens1}
  Let $e\in E$ and let $x,\xz:E\to\RR$ be weightings that only differ
  in the weights given to the element~$e$, and $\B$ be an $x$-optimal
  basis.
  \begin{enumerate}
    \Item{a} If $e\in\B$, then $\B$ is $\xz$-optimal if and
    only if $\xz(e)\leq \minmax{e}{x}$.
    \Item{b} If $e\not\in\B$, then $\B$ is $\xz$-optimal if
    and only if $\xz(e)\geq \minmax{e}{x}$.
    \Item{c} If $|\xz(e)-x(e)|\leq \tol{x}{e}$, then $\B$ is
    $\xz$-optimal.
  \end{enumerate}
\end{thm}

By \cref{thm:sens1}(c), the weight of a single element $e$ can be
changed by $\tol{x}{e}$ while preserving optimality.  In contrast, if
we allow the weights of two or more elements to be changed, then only
changes by at most $\tfrac{1}{2}\tol{x}{e}$ preserve optimality.

\begin{thm}[Sensitivity, global change]\label{thm:sens-gen}
  Let $x:E\to\RR$ be a weighting, and $\B$ be an $x$-optimal basis.
  \begin{enumerate}
    \Item{a} If a weighting $\xz:E\to\RR$ satisfies $|\xz(e)-x(e)|\leq
    \tfrac{1}{2}\,\tol{x}{e}$ for all $e\in E$, then the basis $\B$ is
    $\xz$-optimal.
    \Item{b} For every $\epsilon>0$ there is a weighting $\xz:E\to\RR$
    such that $|\xz(e)-x(e)|\leq \tfrac{1}{2}\,\tol{x}{e}+\epsilon$
    holds for all elements $e\in E$ but the basis $\B$ is not
    $\xz$-optimal.
  \end{enumerate}
\end{thm}
Claim~(b) shows that the upper bound in claim~(a) is tight.

\paragraph{Persistency}
Given a weighting $x:E\to\RR$, the set $E$ of ground elements is split
into three (not necessarily nonempty) subsets:
\begin{enumerate}
\item[] $\allE$ = elements  belonging to all $x$-optimal bases;
\item[] $\noneE$ = elements not belonging to any $x$-optimal basis;
\item[] $\someE$ = elements that belong to some but not to all
    $x$-optimal bases.
\end{enumerate}
The \emph{persistency problem} is to determine this partition.
Elements $e\in \allE\cup\noneE$ are called \emph{persistent}.  Knowing
which ground elements belong to which of these three subsets may be
helpful when constructing an optimal basis. Namely, we can contract
all elements of $\allE$ (that is, include them into the solution),
remove all elements of $\noneE $, and try to extend our partial
solution $\allE$ to an optimal basis by only treating the elements
of~$\someE$.

Cechl\'arov\'a and Lacko~\cite{lacko} characterized the sets $\allE$
and $\noneE$ in terms of the rank function of the underlying matroid:
$e\in \allE$ iff removing $e$ from the set of all elements \emph{not
  heavier} than $e$ decreases the rank of this set, and $e\in \noneE$
iff adding $e$ to the set of all elements \emph{lighter} than $e$
leaves the rank of this set unchanged.  We characterize these sets in
terms of min-max weights of ground elements.

\begin{thm}[Persistency]
  \label{thm:persist}
  Let $x:E\to\RR$ be a weighting, and $e\in E$ be a ground element.
  \begin{enumerate}
    \Item{1} $e\in\allE$ if and only if $\minmax{e}{x}> x(e)$;
    \Item{2} $e\in\noneE$ if and only if $\minmax{e}{x} < x(e)$;
    \Item{3} $e\in\someE$ if and only if $\minmax{e}{x}=x(e)$.
    \Item{4} If all weights are distinct, then $\B=\{e\in E\colon
    \minmax{e}{x} > x(e)\}$ is the unique optimal basis.
  \end{enumerate}
\end{thm}

\paragraph{Organization} In \cref{sec:prelim}, we briefly recall main
matroid concepts and results used in this paper. \Cref{sec:main-lemma}
is devoted to the proof of our main technical tool
(\cref{lemA}). Given \cref{lemA}, the proofs of
\cref{thm:kirchhoff,thm:persist,thm:postopt,thm:sens1,thm:sens-gen}
are fairly simple, and are given in the subsequent
\cref{sec:contr,sec:postopt,sec:sens1,sec:sens-gen,sec:persistency}.

\section{Preliminaries}
\label{sec:prelim}

We use standard matroid terminology as, for example, in Oxley's
book~\cite{oxley}. A \emph{matroid} on a finite set $E$ of
\emph{ground elements} is a pair $\matr=(E,\f)$, where $\f\subseteq
2^E$ is a nonempty downward closed collection of subsets of $E$,
called \emph{independent sets}, with the \emph{augmentation property}:
whenever $I$ and $J$ are independent sets of cardinalities $|I|<|J|$,
there is an element $e\in J\setminus I$ such that the set $I+e$ is
independent; as customary, we abbreviate $I\cup\{e\}$ to $I+e$ and
write $J-e$ for $J\setminus\{e\}$.

\paragraph{Bases and circuits}
An independent set is a \emph{basis} if it is contained in no other
independent set. The augmentation property implies that all bases have
the same cardinality. This property also yields the \emph{basis
  exchange axiom}: if $A$ and $B$ are bases, then for every element
$e\in A$ there is an element $\ee\in B$ such that
$\exchange{A}{e}{\ee}$ is a basis.  The following two important
refinements of the basis exchange axiom are known as the
\emph{symmetric basis exchange} and the \emph{bijective basis
  exchange}.

\begin{prop}[Brualdi~\cite{brualdi}, Brylawski~\cite{brylawski}]\label{prop1}
  Let $\exB$ and $\B$ be bases.
  \begin{enumerate}
    \Item{a} For every $e\in\exB$ there is an $\ee\in \B$ such that
    both $\exchange{\exB}{e}{\ee}$ and $\exchange{\B}{\ee}{e}$ are
    bases.
    \Item{b} There is a bijection $\bij:\exB\to \B$ such that the set
    $\exchange{\exB}{e}{\bij(e)}$ is a basis for every $e\in \exB$.
  \end{enumerate}
\end{prop}

A subset of $E$ is \emph{dependent} if it is not independent. A
\emph{circuit} is a dependent set whose proper subsets are all
independent.  For a ground element $e$, an $e$-\emph{circuit} is a
circuit containing~$e$. An element $e$ is a \emph{loop} if the set
$\{e\}$ is dependent, and is a \emph{coloop} if $e$ belongs to all
bases.  To avoid pathological situations, we assume that our matroid
is \emph{loopless}: no ground element is a loop or a coloop.  We only
need this assumption to ensure two properties: every circuit contains
at least two elements, and for every ground element $e$ at least one
$e$-circuit exists.

\paragraph{Fundamental paths and cuts}
Let $\B$ be a basis, and $e\in E$ a ground element. If $e\not\in\B$,
then the set $\B+e$ must contain at least one $e$-circuit, because
$\B$ is independent but $\B+e$ is dependent. An important fact, shown
by Brualdi~\cite[Lemma~1]{brualdi} (see also
Oxley~\cite[Proposition~1.1.4]{oxley}) and known as the \emph{unique
  circuit property}, is that the set $\B+e$ contains a \emph{unique}
circuit $C$. Since $\B$ is independent, this circuit $C$ is an
$e$-circuit (that is, $C$ contains $e$).  This unique circuit
$C=\Circ{e}{\B}$ is known as the \emph{fundamental circuit of $e$
  relative to~$\B$}. Motivated by graphic matroids
(cf. \cref{ex:graphic}), we call the independent set
\[
\Path{e}{\B}:=\Circ{e}{\B}-e\subseteq\B
\]
the \emph{fundamental path of $e$ relative to~$\B$}.  If $e\in\B$ is a
basis element, then the set
\[
\Cut{e}{\B}:=\{\ee\in E\setminus\B\colon e\in \Circ{\ee}{\B}\}
\]
is known as the \emph{fundamental cut of $e$ relative to~$\B$}. Note
that $e\not\in\Path{e}{\B}$ and $\ee\not\in\Cut{\ee}{\B}$. Also note
the duality: if $e\not\in\B$ and $\ee\in\B$, then
\[
\mbox{$e\in\Cut{\ee}{\B}$ if and only if $\ee\in\Path{e}{\B}$.}
\]

The unique circuit property yields the following equivalent definition
of $\Cut{e}{\B}$ and $\Path{e}{\B}$.

\begin{prop}\label{prop2}
  Let $\B$ be a basis, and $e\in E$ a ground element.
  \begin{enumerate}
    \Item{a} If $e\in \B$, then $\Cut{e}{\B}=\left\{\ee\in
      E\setminus\B\colon \mbox{$\exchange{\B}{e}{\ee}$ is a
        basis}\right\}$.
    \Item{b} If $e\not\in \B$, then
    $\Path{e}{\B}=\left\{\ee\in\B\colon \mbox{$\exchange{\B}{\ee}{e}$
        is a basis}\right\}$.
  \end{enumerate}
\end{prop}

That is, if $e\in\B$, then $\Cut{e}{\B}$ consists of all elements
$\ee\in E\setminus\B$ from outside the basis $\B$ that \emph{can
  replace} $e$ in~$\B$.  If $e\not\in\B$, then $\Path{e}{\B}$ consists
of all basis elements $\ee\in\B$ that \emph{can be replaced} by $e$ in
$\B$.

\begin{proof}
  Claim (a) follows from claim (b) and the aforementioned duality.  To
  show claim (b), let $e\not\in\B$ and $P=\Path{e}{\B}$; hence,
  $P\subseteq\B$. Take an arbitrary element $\ee\in\B$. If $\ee\not\in
  P$, then $\exchange{\B}{\ee}{e}$ cannot be a basis because it
  contains the circuit $P+e$. If $\ee\in P$, then
  $\exB=\exchange{\B}{\ee}{e}$ is a basis because $\ee$ is removed
  from the unique circuit $P+e$ contained in~$\B+e$ (the set $\exB$ is
  independent and has the same cardinality as~$\B$).
\end{proof}

\begin{rem}
  If $e\not\in\B$, then $\Path{e}{\B}$ is nonempty, because $e$ is not
  a loop (the set $\{e\}$ is independent).  If $e\in\B$, then
  \cref{prop2} implies that the set $\Cut{e}{\B}$ is also nonempty.
  Indeed, since $e$ is not a coloop (does not belong to all bases),
  $e\not\in\exB$ holds for some basis $\exB\neq\B$. By the basis
  exchange axiom, there is some element $\ee\in \exB\setminus\B$ such
  that $\exchange{\B}{e}{\ee}$ is a basis. By \cref{prop2}(b), the
  element $\ee$ belongs to the set $\Cut{e}{\B}$; hence,
  $\Cut{e}{\B}\neq\emptyset$.
\end{rem}

\begin{prop}\label{prop3}
  Let $\B$ be a basis, $e\in \B$ and $C$ an $e$-circuit. Then
  $(C-e)\cap\Cut{e}{\B}\neq\emptyset$.
\end{prop}

\begin{proof}
  Let $C$ be an $e$-circuit. Since the set $I=C-e$ is independent, it
  lies in some basis $\exB$, and $e\not\in\exB$ holds since $I+e=C$ is
  already dependent. By \cref{prop1}(a), there is an $\ee\in \exB$
  such that both sets $\exchange{\B}{e}{\ee}$ and
  $\exchange{\exB}{\ee}{e}$ are bases.  By \cref{prop2}, this is
  equivalent to $\ee\in \Cut{e}{\B}$ and $\ee\in\Path{e}{\exB}$.
  Thus, $\Cut{e}{\B}\cap\Path{e}{\exB}\neq\emptyset$.  Since $\exB$ is
  a basis, and both circuits $C$ and $\Path{e}{\exB}+e$ lie in
  $\exB+e$, the uniqueness of fundamental circuits yields
  $\Path{e}{\exB}=C-e$. Hence, $\Cut{e}{\B}\cap (C-e)\neq\emptyset$,
  as claimed.
\end{proof}

\paragraph{Min-max and bottleneck weights of elements}

A \emph{weighting} is an assignment $x:E \to \RR$ of real weights to
the ground elements.  The \emph{weight} of a set $F\subseteq E$ is the
sum $ x(F):=\sum_{\ee\in F}x(\ee) $ of the weights of its
elements. The \emph{minimum weight basis problem} on a matroid
$\matr=(E,\f)$ is, given a weighting $x: E \to \RR$, to determine the
minimum weight of a basis:
\[
\MB{\matr}{x}:=\min_{B~\mathrm{basis}}\ \sum_{\ee\in B}x(\ee)\,.
\]
We call the number $\MB{\matr}{x}$ the \emph{optimal value} (under the
weighting~$x$), and call a basis $\B$ $x$-\emph{optimal} (or just
\emph{optimal}, if the weighting is clear from the context) if
$x(\B)=\MB{\matr}{x}$ holds, that is, if the basis $\B$ is of minimal
$x$-weight.

\begin{rem}
  Note that in the context of the minimum weight basis problem our
  assumption that no ground element $e$ is a loop or a coloop is quite
  natural. If $e$ is a loop, then it belongs to \emph{none} of the
  bases, and the element $e$ contributes nothing to the optimal value
  $\MB{\matr}{x}$. If $e$ is a coloop, then it belongs to \emph{all}
  bases, and the contribution $x(e)$ of the element $e$ to the optimal
  value $\MB{\matr}{x}$ is predetermined.
\end{rem}

The \emph{min-max weight} $\minmax{e}{x}$ of an element $e\in E$ under
a weighting $x:E\to\RR$ is the minimum, over all $e$-circuits $C$, of
the maximum weight of an element in the (independent) set $C-e$:
\[
\minmax{e}{x}:= \min_{C~\ecirc}\ \max_{\ee\in C-e}\ x(\ee)\,.
\]
Since $e$ is not a loop (the set $\{e\}$ is independent), the set
$C-e$ is nonempty for every $e$-circuit~$C$. Moreover, since $e$ is
not a coloop, at least one $e$-circuit $C$ exists. So, the min-max
weight is well-defined.  Note that the min-max weight $\minmax{e}{x}$
of $e$ does not depend on the weight $x(e)$ of the element $e$ itself:
it only depends on the weights of the remaining elements. So, all
three relations $\minmax{e}{x}<x(e)$, $\minmax{e}{x}=x(e)$ and
$\minmax{e}{x}>x(e)$ are possible. We call
\[
\bottl{e}{x}:=\min\{x(e),\minmax{e}{x}\}
\]
the \emph{bottleneck weight} of~$e$.

\begin{rem}\label{rem:edmonds}
  Thanks to a classical ``bottleneck extrema'' result of Edmonds and
  Fulkerson~\cite{edmonds-fulkerson}, the min-max weight
  $\minmax{e}{x}$ of a ground element $e\in E$ has an equivalent
  definition as the ``max-min'' weight in terms of cocircuits:
  \begin{equation}\label{eq:max-min}
    \minmax{e}{x}=\max_{D~\ecocirc}\ \min_{\ee\in D-e}\ x(\ee)\,.
  \end{equation}
  A set $D\subseteq E$ is a \emph{cocircuit} if it intersects every
  basis, and no proper subset has this property.  In other words,
  cocircuits in a matroid $\matr$ are circuits of the dual matroid
  $\matr^*$; bases in $\matr^*$ are complements of the bases of
  $\matr$. An $e$-\emph{cocircuit} is a cocircuit containing the
  element~$e$. To obtain \cref{eq:max-min} from the main theorem of
  Edmonds and Fulkerson~\cite{edmonds-fulkerson}, one has only to
  verify that the family $\hm=\{D-e\colon \mbox{$D$ is an
    $e$-cocircuit}\}$ is a blocking family for the family
  $\fm=\{C-e\colon \mbox{$C$ is an $e$-circuit}\}$, i.e., that members
  of $\hm$ are minimal (under set inclusion) subsets of $E$
  intersecting all members of $\fm$. This follows from a well-known
  fact that $|C\cap D|\neq 1$ holds for every circuit $C$ and every
  cocircuit~$D$ (see, for
  example~\cite[Proposition~2.1.11]{oxley}). In this paper, we will
  not use this equivalent ``max-min'' definition \cref{eq:max-min} of
  $\minmax{e}{x}$: we only mention it for interested readers.
\end{rem}

\begin{ex}
  \label{ex:graphic}
  The \emph{graphic matroid} (or \emph{cycle matroid}) $\matr(G)$
  determined by an undirected connected graph $G=(V,E)$ has edges of
  $G$ as its ground elements. Independent sets are forests, bases are
  spanning trees of $G$, and circuits are simple cycles in~$G$. A loop
  is an edge with identical endpoints, and a coloop is an edge $e$
  whose deletion destroys the connectivity of $G$ (such edges are also
  called \emph{bridges}).  The min-max weight $\minmax{e}{x}$ of an
  edge $e\in E$ is the minimum, over all simple paths in $G$ of length
  at least two between the endpoints of $e$, of the maximum weight of
  an edge in this path.  The bottleneck weight
  $\bottl{e}{x}=\min\left\{x(e),\minmax{e}{x}\right\}$ of an edge $e$
  is also known as the \emph{bottleneck distance} between the
  endpoints of~$e$. If $T$ is a spanning tree of $ G$ and $e\not\in T$
  is an edge of $G$, then the set $\Path{e}{T}$ consists of all edges
  of the unique path in $T$ between the endpoints of $e$. If $\ee\in
  T$, then $\Cut{\ee}{T}$ consists  of all edges of $G$ lying between
  the two trees of $T-\ee$, except the edge $\ee$ itself.
\end{ex}

For the rest of the paper, let $\matr=(E,\f)$ be an arbitrary loopless
matroid.

\section{Main lemma}
\label{sec:main-lemma}

The following lemma (illustrated in \cref{figA}) is our main technical
tool.

\begin{lem}[Main lemma]
  \label{lemA}
  Let $x:E\to\RR$ be a weighting, $e\in E$ a ground element, and $\B$
  an $x$-optimal basis.
  \begin{enumerate}
    \Item{a} If $e\in\B$, then
    \begin{itemize}
    \item $x(e)\leq \minmax{e}{x} =$ the minimum weight of an element
      in $\Cut{e}{\B}$;
    \item the minimum weight of a basis avoiding the element $e$ is
      $x(\B)-x(e)+\minmax{e}{x}$.
    \end{itemize}

    \Item{b} If $e\not\in\B$, then
    \begin{itemize}
    \item $x(e)\geq \minmax{e}{x} =$ the maximum weight of an element
      in $\Path{e}{\B}$;
    \item the minimum weight of a basis containing the element $e$ is
      $x(\B)-\minmax{e}{x}+x(e)$.
    \end{itemize}
  \end{enumerate}
\end{lem}

\begin{figure}
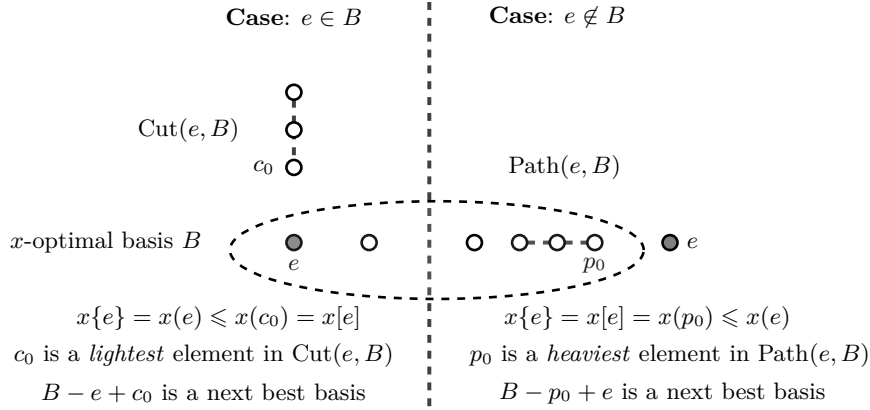

  \pathscuts
  \caption{A schematic summary of \cref{lemA}. A \emph{lightest}
    (resp., \emph{heaviest}) element of a set $S\subseteq E$ is an
    element of $S$ of the minimum (resp., maximum) weight (there may
    be several such elements). If $e\in\B$, then a \emph{next best}
    basis is a basis which has the minimal weight under all bases
    \emph{avoiding} the element $e$. If $e\not\in\B$, then a
    \emph{next best} basis is a basis which has the minimal weight
    under all bases \emph{containing} the element~$e$.}
  \label{figA}
\end{figure}

Given a weighting $x:E\to\RR$ and an element $e\in E$, by an
$e$-circuit \emph{witnessing} the min-max weight $\minmax{e}{x}$ of an
element $e$ we will mean an $e$-circuit $C$ on which the min-max
weight of the element $e$ is achieved, that is, for which
$\minmax{e}{x}$ is the maximum weight $x(\ee)$ of an element $\ee\in
C-e$.

\begin{proof}[Proof of \cref{lemA}(a)]
  Let $x:E\to\RR$ be a weighting, $e\in E$ a ground element, and $\B$
  an optimal basis. Assume that $e\in\B$, and let $\eecut$ be a
  lightest element in $\Cut{e}{\B}$.  Our goal is to show that

  \begin{enumerate}
    \Item{i} $x(e)\leq \minmax{e}{x} =x(\eecut)$, and
    \Item{ii} the set $\exchange{\B}{e}{\eecut}$ is a lightest basis
    among all bases avoiding the element~$e$.
  \end{enumerate}

  \noindent (i) To show the inequality $x(e)\leq x(\eecut)$, suppose
  for a contradiction that $x(\eecut) < x(e)$ holds. Since
  $\eecut\in\Cut{e}{\B}$, \cref{prop2} implies that the set
  $\exchange{\B}{e}{\eecut}$ is a basis. But then its weight is
  smaller than that of $\B$, contradicting the optimality of~$\B$.

  To show the inequality $\minmax{e}{x}\geq x(\eecut)$, let $C$ be an
  $e$-circuit witnessing the min-max weight $\minmax{e}{x}$ of the
  element~$e$. Hence, $\minmax{e}{x}=x(\ee_0)$, where $\ee_0$ is a
  \emph{heaviest} element of $C-e$. By \cref{prop3}, there is an
  element $\eee$ in the intersection $(C-e)\cap\Cut{e}{\B}$.  Then
  $x(\eee)\leq x(\ee_0)= \minmax{e}{x}$ because $\eee\in C-e$ and
  $\ee_0$ is a heaviest element of $C-e$, and $x(\eee)\geq x(\eecut)$
  because $\eee\in\Cut{e}{\B}$ and $\eecut$ is a lightest element of
  $\Cut{e}{\B}$. Hence, $\minmax{e}{x}\geq x(\eecut)$.

  To show the opposite inequality $\minmax{e}{x}\leq x(\eecut)$,
  consider the fundamental circuit $C=\Path{\eecut}{\B}+\eecut$ of the
  element $\eecut$ relative to the basis $\B$.  Since
  $\eecut\in\Cut{e}{\B}$, we have $e\in \Path{\eecut}{\B}$. Thus, both
  $e$ and $\eecut$ belong to the same circuit $C$. Let $\eepath$ be a
  \emph{heaviest} element in $C-e=\Path{\eecut}{\B}+\eecut-e$.  Since
  in the definition of the min-max weight $\minmax{e}{x}$ we take the
  \emph{minimum} over all circuits containing $e$, we have
  $\minmax{e}{x}\leq x(\eepath)$. So, it remains to show that
  $x(\eepath)\leq x(\eecut)$ holds.  Suppose for a contradiction that
  we have a strict inequality $x(\eepath)>x(\eecut)$. Then (clearly)
  $\eepath\neq\eecut$ and, hence, $\eepath\in \Path{\eecut}{\B}$. By
  \cref{prop2}, the set $\exB=\exchange{\B}{\eepath}{\eecut}$ is a
  basis. But the weight of this basis is
  $x(\exB)=x(\B)-x(\eepath)+x(\eecut)<x(\B)$, contradicting the
  optimality of the basis~$\B$. Thus, $\minmax{e}{x}\leq x(\eecut)$,
  as desired.

  (ii) Let $\nbases$ be the family of all bases \emph{avoiding} the
  element $e$; hence, $\B\not\in \nbases$.  Since $\eecut$ belongs to
  $\Cut{e}{\B}$, \cref{prop2} implies that the set
  $\exchange{\B}{e}{\eecut}$ is a basis, and this basis belongs to
  $\nbases$. Our goal is to show that this specific basis has the
  smallest weight among all bases in $\nbases$.

  So, let $\xB$ be a lightest basis in $\nbases$; hence,
  $e\in\B\setminus\xB$. By \cref{prop1}(b), there is a bijection
  $\bij:\B\to\xB$ such that the set $\exchange{\B}{\ee}{\bij(\ee)}$ is
  a basis for every element $\ee\in\B$. Since the basis $\B$ is
  optimal, this yields $x(\ee)\leq x(\bij(\ee))$ for every
  $\ee\in\B$. Consider the basis $\exB=\exchange{\B}{e}{c}$ where
  $c:=\bij(e)\in\xB$. Its weight is
  \[
  x(\exB)=x(c)+\sum_{\ee\in\B-e}x(\ee)\leq
  x(c)+\sum_{\ee\in\B-e}x(\bij(\ee))= x(c)+x(\xB-c)=x(\xB)\,.
  \]
  Since $\exB$ is a basis, \cref{prop2} implies that
  $c\in\Cut{e}{\B}$.  Since both elements $c$ and $\eecut$ belong to
  $\Cut{e}{\B}$, and since $\eecut$ is a lightest element of
  $\Cut{e}{\B}$, we have $x(\eecut)\leq x(c)$. So
  $x(\exchange{\B}{e}{\eecut})\leq x(\exchange{\B}{e}{c})= x(\exB)\leq
  x(\xB)$, meaning that $\exchange{\B}{e}{\eecut}$ is a lightest basis
  in $\nbases$, as claimed.
\end{proof}

\begin{rem}
  The inequality $\minmax{e}{x}\geq x(\eecut)$ in part (i) also
  follows from the equivalent definition \cref{eq:max-min} of the
  min-max weight $\minmax{e}{x}$ (see \cref{rem:edmonds}). For this,
  it is enough to verify that the set $\Cut{e}{\B}+e$ is a
  cocircuit. The set $\Cut{e}{\B}$ intersects every basis $\exB$ with
  $e\not\in\exB$: by the basis exchange axiom, $\exchange{\B}{e}{a}$
  is a basis for some $a\in \exB$; hence, $a\in
  \Cut{e}{\B}$. Moreover, no proper subset of $\Cut{e}{\B}$ has this
  property: for every $\ee\in \Cut{e}{\B}$, the set $\Cut{e}{\B}-\ee$
  does not intersect the basis $\exchange{\B}{e}{\ee}$.
\end{rem}

\begin{proof}[Proof of \cref{lemA}(b)]
  Let $x:E\to\RR$ be a weighting, $e\in E$ a ground element, and $\B$
  an optimal basis. Assume that $e\not \in\B$, and let $\eepath$ be a
  heaviest element in $\Path{e}{\B}$. Our goal is to show that
  \begin{enumerate}
    \Item{i} $x(e)\geq \minmax{e}{x} = x(\eepath)$, and
    \Item{ii} the set $\exchange{\B}{\eepath}{e}$ is a lightest basis
    among all bases containing the element $e$.
  \end{enumerate}

  \noindent (i) To show the inequality $x(e)\geq x(\eepath)$, suppose
  for a contradiction that $x(e)< x(\eepath)$. Since
  $\eepath\in\Path{e}{\B}$, \cref{prop2} implies that the set
  $\exchange{\B}{\eepath}{e}$ is a basis. But then its weight is
  smaller than that of $\B$, contradicting the optimality of~$\B$.

  The inequality $\minmax{e}{x} \leq x(\eepath)$ holds because
  $\Path{e}{\B}+e$ is an $e$-circuit, and $\minmax{e}{x}$ takes the
  \emph{minimum} (of the maximum weights) over \emph{all}
  $e$-circuits.  To show the opposite inequality $\minmax{e}{x} \geq
  x(\eepath)$, suppose for a contradiction that we have a strict
  inequality $\minmax{e}{x} < x(\eepath)$, and let $C$ be an
  $e$-circuit witnessing $\minmax{e}{x}$.  Hence, $x(\ee) <
  x(\eepath)$ holds for all $\ee\in C-e$.  Since
  $\eepath\in\Path{e}{\B}$, \cref{prop2} implies that
  $\exB=\exchange{\B}{\eepath}{e}$ is also a basis. Since $e\in\exB$
  and $C$ is an $e$-circuit, \cref{prop3} implies that some element
  $\ee_0\in C-e$ belongs to $\Cut{e}{\exB}$. So, by \cref{prop2}, the
  set $\exB'=\exchange{\exB}{e}{\ee_0}=\exchange{\B}{\eepath}{\ee_0}$
  is a basis. But since $x(\ee_0) < x(\eepath)$, we have $x(\exB') <
  x(\B)$, contradicting the optimality of~$\B$.

  \medskip
  \noindent (ii) Let $\pbases$ be the family of all bases
  \emph{containing} the element $e$; hence, $\B\not\in \pbases$.
  Since $\eepath$ belongs to $\Path{e}{\B}$, \cref{prop2} implies that
  the set $\exchange{\B}{\eepath}{e}$ is a basis, and this basis
  belongs to $\pbases$. Our goal is to show that this specific basis
  has the smallest weight among all bases in $\pbases$.

  So, let $\xB$ be a lightest basis in $\pbases$; hence,
  $e\in\xB\setminus\B$. By \cref{prop1}(b), there is a bijection
  $\bij:\xB\to\B$ such that the set $\exchange{\B}{\bij(\ee)}{\ee}$ is
  a basis for every element $\ee\in\xB$. Since the basis $\B$ is
  optimal, this yields $x(\bij(\ee))\leq x(\ee)$ for every
  $\ee\in\xB$. Consider the basis $\exB=\exchange{\B}{p}{e}$ where
  $p:=\bij(e)\in\B$. Its weight is
  \[
  x(\exB)=x(e)+\sum_{\ee\in\xB-e}x(\bij(\ee))\leq
  x(e)+\sum_{\ee\in\xB-e}x(\ee)= x(e)+x(\xB-e)=x(\xB)\,.
  \]
  Since $\exB$ is a basis, \cref{prop2} implies that
  $p\in\Path{e}{\B}$.  Since both elements $p$ and $\eepath$ belong to
  $\Path{e}{\B}$, and since $\eepath$ is a \emph{heaviest} element of
  $\Path{e}{\B}$, we have $x(\eepath)\geq x(p)$. So
  $x(\exchange{\B}{\eepath}{e})\leq x(\exchange{\B}{p}{e})=x(\exB)\leq
  x(\xB)$, meaning hat $\exchange{\B}{\eepath}{e}$ is a lightest basis
  in $\pbases$, as claimed.
\end{proof}

\begin{rem}[From optimal bases to bottleneck weights]\label{rem:basis-to-minmax}
  Having an optimal basis $\B$, we can determine the bottleneck weight
  $\bottl{e}{x}=\min\{x(e), \minmax{e}{x}\}$ of any ground element
  $e\in E$ from the weights of elements of $\B$: if $e\in\B$, then
  $\bottl{e}{x}=x(e)$ (by \cref{lemA}(a)), and if $e\not\in\B$, then
  $\bottl{e}{x}=\minmax{e}{x}$ is the weight of a heaviest element in
  $\Path{e}{\B}$ (by \cref{lemA}(b)). The following proposition shows
  how to find such a heaviest element in $\Path{e}{\B}$.
\end{rem}

\begin{prop}
  Let $x:E\to\RR$ be a weighting and $\B=\{\ee_1,\ldots,\ee_r\}$ be an
  optimal basis with $x(\ee_1)\leq\ldots\leq x(\ee_r)$. If
  $e\not\in\B$, then $\minmax{e}{x}=x(\ee_i)$, where $i$ is the
  smallest index for which the set $\{\ee_1,\ldots,\ee_i,e\}$ is
  dependent.
\end{prop}

\begin{proof}
  For $j=1,\ldots,r$, let $\B_j=\{\ee_1,\ldots,\ee_j\}$ be the set of
  the $j$ lightest elements of $\B$, and let $\B_0=\emptyset$. The set
  $\B_0+e=\{e\}$ is independent because $e$ is not a loop, and the set
  $\B_r+e=\B+e$ is dependent, because $\B$ is a basis and
  $e\not\in\B$. So, there is a unique index $i\in\{1,\ldots,r\}$ such
  that the set $\B_{i-1}+e$ is independent but $\B_{i}+e$ is
  dependent.  Our goal is to show that $\minmax{e}{x}=x(\ee_{i})$
  holds for this~$i$.

  Since $\B_i$ is independent but $\B_i+e$ is dependent, the set
  $\B_i+e$ contains an $e$-circuit $C$. Since $C\subseteq
  \B_i+e\subseteq \B+e$, the uniqueness of fundamental circuits yields
  $C=\Circ{e}{\B}$; hence, $\Path{e}{\B}=\Circ{e}{\B}-e\subseteq
  \B_{i}$. Since the set $\B_{i-1}+e$ is independent, the last element
  $\ee_i$ of $\B_i$ must be contained in $\Path{e}{\B}$.  Since
  $\Path{e}{\B}\subseteq \B_i$ and since $\ee_i$ is a heaviest element
  of $\B_i$, $\ee_i$ is also a heaviest element of
  $\Path{e}{\B}$. Thus, \cref{lemA}(b) gives $\minmax{e}{x}
  =x(\ee_i)$.
\end{proof}

\section{Proof of Theorem~\ref{thm:kirchhoff}}
\label{sec:contr}

Let, as before, $\matr=(E,\f)$ be a loopless matroid, and $e\in E$ be
a ground element. Recall that the independent sets of the matroid
$\contr{\matr}{e}$, obtained by \emph{contracting} the element $e$,
are all sets $I-e$ with $I\in\f$ and $e\in I$, while those of the
matroid $\remove{\matr}{e}$, obtained by \emph{deleting} the element
$e$, are all sets $I\in\f$ with $e\not\in I$.  Our goal is to show
that, for every weighting $x:E\to\RR$, the following equalities hold:
\begin{enumerate}
  \Item{a} $\MB{\contr{\matr}{e}}{x} =\MB{\matr}{x} - \bottl{e}{x}$;
  \Item{b} $\MB{\fremove{\matr}{e}}{x}
  =\MB{\matr}{x}-\bottl{e}{x}+\minmax{e}{x}$.
\end{enumerate}

Take an arbitrary optimal basis $\B$ of $\matr$; hence,
$\MB{\matr}{x}=x(\B)$. In the proof of both equalities (a) and (b), we
distinguish two cases depending on whether our element $e$ belongs to
$\B$ or not.

(a) If $e\in\B$, then \cref{lemA}(a) yields $\bottl{e}{x}=x(e)$ and,
since then $\B-e$ is an optimal basis of $\contr{\matr}{e}$, we obtain
$\MB{\contr{\matr}{e}}{x}=x(\B)-x(e)=\MB{\matr}{x}-\bottl{e}{x}$.  If
$e\not\in\B$, then consider a basis $\exB$ of minimum weight among all
bases of $\matr$ containing the element~$e$. By \cref{lemA}(b), we
have $\bottl{e}{x}=\minmax{e}{x}$ and
$x(\exB)=x(\B)-\minmax{e}{x}+x(e)$. Since then $\exB-e$ is an optimal
basis of $\contr{\matr}{e}$, we obtain $
\MB{\contr{\matr}{e}}{x}=x(\exB-e)
=x(\B)-\minmax{e}{x}=\MB{\matr}{x}-\bottl{e}{x}$.

(b) If $e\not\in\B$, then \cref{lemA}(b) yields
$\bottl{e}{x}=\minmax{e}{x}$ and, since then $\B$ is also an optimal
basis of $\remove{\matr}{e}$, we obtain
$\MB{\fremove{\matr}{e}}{x}=x(\B)=\MB{\matr}{x}
=\MB{\matr}{x}-\bottl{e}{x}+\minmax{e}{x}$. If $e\in\B$, then consider
a basis $\exB$ of minimum weight among all bases of $\matr$ avoiding
the element~$e$. By \cref{lemA}(a), we have $\bottl{e}{x}=x(e)$ and
$x(\exB)=x(\B)-x(e)+\minmax{e}{x}$.  Since $\exB$ is an optimal basis
of $\remove{\matr}{e}$, we obtain
$\MB{\fremove{\matr}{e}}{x}=x(\exB)=x(\B)-x(e)+\minmax{e}{x}
=\MB{\matr}{x}-\bottl{e}{x}+\minmax{e}{x}$.  \qed

\section{Proof of Theorem~\ref{thm:postopt}}
\label{sec:postopt}
Let $e\in E$ and let $x,\xz:E\to\RR$ be weightings that differ only in
the weights given to the element~$e$. Our goal is to show the equality
$\MB{\matr}{\xz}-\MB{\matr}{x}=\bottl{e}{\xz}-\bottl{e}{x}$.

Recall that the independent sets of the matroid $\contr{\matr}{e}$ are
all sets $I-e$ with $I\in\f$ and $e\in I$. By \cref{thm:kirchhoff}(a),
the equality $\MB{\contr{\matr}{e}}{z} =\MB{\matr}{z} - \bottl{e}{z}$
holds for every weighting $z:E\to\RR$.  Since the weighting $\xz$ does
not change the weight of elements in $E-e$, we have
$\MB{\contr{\matr}{e}}{\xz}=\MB{\contr{\matr}{e}}{x}$. So,
\cref{thm:kirchhoff}(a) yields
\[
\MB{\matr}{\xz} - \bottl{e}{\xz} =
\MB{\contr{\matr}{e}}{\xz}=\MB{\contr{\matr}{e}}{x} = \MB{\matr}{x} -
\bottl{e}{x}\,,
\]
from which $\MB{\matr}{\xz}-
\MB{\matr}{x}=\bottl{e}{\xz}-\bottl{e}{x}$ follows.  \qed

\section{Proof of Theorem~\ref{thm:sens1}}
\label{sec:sens1}
Fix a ground element $e\in E$, and let $x,\xz:E\to\RR$ be weightings
that only differ in the weights given to~$e$. Since the min-max weight
of $e$ only depends on the weights of the elements in $E-e$, and since
the weighting $\xz$ leaves these weights unchanged, we have
$\minmax{e}{\xz}=\minmax{e}{x}$, that is, the min-max weight of the
element $e$ does not change. Thus, the bottleneck weight of $e$ under
the new weighting $\xz$ is
$\bottl{e}{\xz}=\min\{\xz(e),\minmax{e}{x}\}$.  Recall that the
\emph{tolerance} of the element $e$ under the weighting $x$ is
$\tol{x}{e}=|\minmax{e}{x}-x(e)|$.

Let $\B$ be an $x$-optimal basis. Our goal is to prove the following
three assertions.
\begin{enumerate}
  \Item{a} If $e\in\B$, then $\B$ is $\xz$-optimal if and only if
  $\xz(e)\leq \minmax{e}{x}$.
  \Item{b} If $e\not\in\B$, then $\B$ is $\xz$-optimal if and only if
  $\xz(e)\geq \minmax{e}{x}$.
  \Item{c} If $|\xz(e)-x(e)|\leq \tol{x}{e}$, then  $\B$ is
  $\xz$-optimal.
\end{enumerate}
\begin{proof}
  (a) Let $e\in\B$. Then $\xz(B) = x(B) + \xz(e) - x(e)$ and, by
  \cref{lemA}(a), $\bottl{e}{x}=x(e)$.  \Cref{thm:postopt} yields $
  \MB{\matr}{\xz}=\MB{\matr}{x} + \bottl{e}{\xz}-\bottl{e}{x} = x(B) +
  \min\{\xz(e), \minmax{e}{x}\}-x(e)\,.  $ The basis $\B$ is
  $\xz$-optimal iff $\xz(B) =\MB{\matr}{\xz}$, which happens precisely
  when $\min\{\xz(e), \minmax{e}{x}\}=\xz(e)$, that is, when
  $\xz(e)\leq \minmax{e}{x}$.

  (b) Let $e\not\in\B$. Then $\xz(B)=x(B)$ and, by \cref{lemA}(b),
  $\bottl{e}{x}=\minmax{e}{x}$.  \Cref{thm:postopt} yields $
  \MB{\matr}{\xz}=\MB{\matr}{x} + \bottl{e}{\xz}-\bottl{e}{x} = x(B) +
  \min\{\xz(e), \minmax{e}{x}\}-\minmax{e}{x}\,.  $ The basis $\B$ is
  $\xz$-optimal iff $\xz(B) =\MB{\matr}{\xz}$, which happens precisely
  when $\min\{\xz(e), \minmax{e}{x}\}=\minmax{e}{x}$, that is, when
  $\xz(e)\geq \minmax{e}{x}$.

  (c) Assume $|\xz(e)-x(e)|\leq \tol{x}{e}$, i.e.,
  $x(e)-\tol{x}{e}\leq \xz(e)\leq x(e)+\tol{x}{e}$.  If $e\in\B$, then
  \cref{lemA}(a) implies $x(e)\leq\minmax{e}{x}$ and, hence,
  $\tol{x}{e}=\minmax{e}{x}-x(e)$.  Thus, $\xz(e)\leq
  x(e)+\tol{x}{e}=\minmax{e}{x}$, and claim (a) ensures that the basis
  $\B$ is $\xz$-optimal.  If $e\not\in\B$, then \cref{lemA}(b) implies
  $x(e)\geq\minmax{e}{x}$ and, hence,
  $\tol{x}{e}=x(e)-\minmax{e}{x}$. Thus, $\xz(e)\geq
  x(e)-\tol{x}{e}=\minmax{e}{x}$, and claim (b) ensures that the basis
  $\B$ is $\xz$-optimal.
\end{proof}

\begin{rem}\label{rem:libura}
  Given \cref{lemA}, claims (a) and (b) of \cref{thm:sens1} also
  follow from a result of Libura~\cite[Lemma~4]{libura} stating that
  $\B$ is $\xz$-optimal iff $\xz(e)\leq x(\eecut)$ holds for a
  lightest element $\eecut$ of $\Cut{e}{\B}$ (when $e\in\B$) or
  $\xz(e)\geq x(\eepath)$ holds for a heaviest element $\eepath$ in
  $\Path{e}{\B}$ (when $e\not\in\B$).  By \cref{lemA},
  $\minmax{e}{x}=x(\eecut)$ (when $e\in\B$) and
  $\minmax{e}{x}=x(\eepath)$ (when $e\not\in\B$).
\end{rem}

\section{Proof of Theorem~\ref{thm:sens-gen}}
\label{sec:sens-gen}

We will need the following simple fact.

\begin{prop}\label{prop4}
  Let $x:E\to\RR$ be a weighting, and $\B$ a basis. If $\B$ is not
  $x$-optimal, then $x(e)>x(\ee)$ holds for some elements $e\in\B$ and
  $\ee\in\Cut{e}{\B}$.
\end{prop}

\begin{proof}
  Let $\exB$ be an $x$-optimal basis; hence, $x(\B)>x(\exB)$. By
  \cref{prop1}(b), there is a bijection $\bij:\B\to \exB$ such that
  the set $\exchange{\B}{e}{\bij(e)}$ is a basis for every $e\in
  \B$. Hence, by \cref{prop2}, $\bij(e)\in\Cut{e}{\B}$ holds for every
  $e\in\B\setminus\exB$.  Finally, since $
  \sum_{e\in\B}x(e)=x(\B)>x(\exB)=\sum_{e\in\B}x(\bij(e))\,, $ a
  strict inequality $x(e)>x(\bij(e))$ must hold for at least one
  element~$e\in\B$.
\end{proof}

\begin{proof}[Proof of \cref{thm:sens-gen}(a)]
  Let $x:E\to\RR$ be a weighting, and let $\xz:E\to\RR$ be a weighting
  satisfying $|\xz(e)-x(e)|\leq\tfrac{1}{2}\,\tol{x}{e}$ for all $e
  \in E$.  Our goal is to show that then every $x$-optimal basis is
  also $\xz$-optimal.

  Assume to the contrary that some $x$-optimal basis $\B$ is not
  $\xz$-optimal. Then, by \cref{prop4}, $\xz(\ee) < \xz(e)$ holds for
  some elements $e\in\B$ and $\ee\in \Cut{e}{\B}$; hence, we also have
  $e\in \Path{\ee}{\B}$.  Since $e\in\B$ and $\ee \in \Cut{e}{\B}$,
  \cref{lemA}(a) yields $x(e)\leq \minmax{e}{x}\leq x(\ee)$. Since
  $\ee\not\in\B$ and $e\in \Path{\ee}{\B}$, \cref{lemA}(b) yields
  $x(\ee)\geq \minmax{\ee}{x}\geq x(e)$. In particular,
  $\tol{x}{e}=\minmax{e}{x}-x(e)$ and $\tol{x}{\ee} =
  x(\ee)-\minmax{\ee}{x}$.  Putting everything together, we get
  \[
  \frac{\minmax{e}{x}+x(e)}{2} \leq \frac{x(\ee)+\minmax{\ee}{x}}{2} =
  x(\ee) - \frac{\tol{x}{\ee}}{2} \leq \xz(\ee) < \xz(e) \leq x(e) +
  \frac{ \tol{x}{e} }{2} = \frac{x(e)+\minmax{e}{x}}{2}\,,
  \]
  a contradiction.
\end{proof}

\begin{proof}[Proof of \cref{thm:sens-gen}(b)]
  Let $\epsilon>0$ and let $x:E\to\RR$ be a weighting. Take an
  arbitrary $x$-optimal basis $\B$. Our goal is to show that there is
  a weighting $\xz:E\to\RR$ such that $|\xz(e)-x(e)|\leq
  \tfrac{1}{2}\,\tol{x}{e}+\epsilon$ holds for all elements $e\in E$
  but the basis $\B$ is not $\xz$-optimal.

  Consider all pairs $(e,\ee)$ such that $e\in\B$ and
  $\ee\in\Cut{e}{\B}$; hence, $e\in\Path{\ee}{\B}$. Since the basis
  $\B$ is $x$-optimal, \cref{prop2} implies that $x(\ee)\geq x(e)$
  holds for every such pair.  So, let $(e,\ee)$ be a pair for which
  the difference $x(f)-x(e)$ is smallest possible.  Then $\ee$ is a
  lightest element in $\Cut{e}{\B}$ and $e$ is a heaviest element in
  $\Path{\ee}{\B}$. By \cref{lemA}, $x(e)\leq \minmax{e}{x}=x(\ee)$
  and $x(\ee)\geq \minmax{\ee}{x}=x(e)$. Hence,
  $\tol{x}{e}=\minmax{e}{x}-x(e)=x(\ee)-x(e)=x(\ee)-\minmax{\ee}{x}=\tol{x}{\ee}$,
  that is, both elements $e$ and $\ee$ have the same tolerance
  $t:=\tol{x}{e}=\tol{x}{\ee}$ under the weighting~$x$.

  Now, let $\xz:E\to\RR$ be the weighting with
  $\xz(e):=x(e)+\tfrac{1}{2}\,t+\epsilon$,
  $\xz(\ee):=x(\ee)-\tfrac{1}{2}\,t-\epsilon$, and $\xz(\eee):=
  x(\eee)$ for all other elements $\eee$.  So,
  $|\xz(\eee)-x(\eee)|=\tfrac{1}{2}\,\tol{x}{\eee}+\epsilon$ for
  $\eee\in\{e,\ee\}$, and $|\xz(\eee)-x(\eee)|=0 <
  \tfrac{1}{2}\,\tol{x}{\eee}+\epsilon$ for all
  $\eee\not\in\{e,\ee\}$.  Then
  $\xz(e)-\xz(\ee)=x(e)-x(\ee)+t+2\epsilon=2\epsilon>0$, and
  $\xz(\ee)<\xz(e)$ implies that the basis $\exchange{\B}{e}{\ee}$ has
  smaller $\xz$-weight than $\B$, so $\B$ cannot be $\xz$-optimal.
\end{proof}

\section{Proof of Theorem~\ref{thm:persist}}
\label{sec:persistency}

Every weighting $x:E\to\RR$ yields the partition $E=\allE \cup \noneE
\cup \someE$ of ground elements into three (not necessarily nonempty)
subsets $\allE$ (elements belonging to all $x$-optimal bases),
$\noneE$ (elements not belonging to any $x$-optimal basis), and
$\someE$ (elements belonging to some but not to all $x$-optimal
bases). Our goal is to prove the following claims:

\begin{enumerate}
  \Item{1} $e\in\allE$ if and only if $\minmax{e}{x}> x(e)$;
  \Item{2} $e\in\noneE$ if and only if $\minmax{e}{x} < x(e)$;
  \Item{3} $e\in\someE$ if and only if $\minmax{e}{x}=x(e)$.
  \Item{4} If all weights are distinct, then $\B=\{e\in E\colon
  \minmax{e}{x} > x(e)\}$ is the unique optimal basis.
\end{enumerate}

\begin{proof}
  (1) To show the direction $(\Rightarrow)$, let $e\in\allE$ and take
  any optimal basis $\B$; hence, $e\in \B$. By \cref{lemA}(a), we then
  have $x(e)\leq \minmax{e}{x}=x(\eecut)$, where $\eecut$ is a
  lightest element in $\Cut{e}{\B}$. By \cref{prop2}, the set
  $\exB=\exchange{\B}{e}{\eecut}$ is a basis. If the equality
  $x(e)=\minmax{e}{x}$ held, then this basis would be optimal,
  too. But $e\not\in\exB$, a contradiction with $e\in\allE$. Hence
  $\minmax{e}{x}> x(e)$ holds. The opposite direction $(\Leftarrow)$
  follows directly from \cref{lemA}(b): if the element $e$ is avoided
  by some optimal basis, then $\minmax{e}{x}\leq x(e)$ holds.

  (2) The proof of this claim is similar. To show the direction
  $(\Rightarrow)$, let $e\in\noneE$ and take any optimal basis $\B$;
  hence, $e\not\in \B$. By \cref{lemA}(b), we then have $x(e)\geq
  \minmax{e}{x}=x(\eepath)$, where $\eepath\in\B$ is a heaviest
  element in $\Path{e}{\B}$. By \cref{prop2}, the set
  $\exB=\exchange{\B}{\eepath}{e}$ is a basis. If the equality
  $x(e)=\minmax{e}{x}$ held, then this basis would be optimal,
  too. But $e\in\exB$, a contradiction with $e\in\noneE$. Hence,
  $\minmax{e}{x} < x(e)$ holds. The opposite direction $(\Leftarrow)$
  in (2) follows directly from \cref{lemA}(a): if the element $e$ is
  contained in some optimal basis, then $\minmax{e}{x} \geq x(e)$
  holds.

  (3) Follows directly from claims (1) and (2).

  (4) Assume that all weights are distinct.  Then the optimal basis
  $\B$ is \emph{unique}: if there were two distinct optimal bases,
  then (by the basis exchange axiom) a heaviest element, lying in one
  basis but not in the other, could be replaced by a (strictly)
  lighter element of the other basis, contradicting the optimality of
  the former basis.  Since the basis $\B$ is unique, we have
  $\B=\allE$ and, by (1), $\B=\{e\in E\colon \minmax{e}{x} > x(e)\}$,
  as claimed.
\end{proof}

\end{document}